\documentclass[a4paper,12pt]{article}
\usepackage{makeidx}
\usepackage{graphicx}
\usepackage{multicol}
\usepackage[bottom]{footmisc}
\usepackage{amsmath}
\usepackage{amsfonts}
\usepackage{bbm}
\usepackage{amsthm}

\makeindex
\newtheorem{theorem}{Theorem}[subsection]

\def \eps {{\epsilon}}
\def \sig {{\sigma}}
\def \D {{\rm d}}

\def \hP {{\widehat{P}}}
\def \hR {{\widehat{R}}}
\def \hS {{\widehat{S}}}

\def \hY {{\widehat{Y}}}

\def \hp {{\widehat{p}}}
\def \hN {{\widehat{N}}}
\def \htau {{\widehat{\tau}}}

\def \VV {{\mathbb{V}}}
\def \EE {{\mathbb{E}}}
\def \ii {{\mathbbm{1}}}
\begin{document}

\title{Multilevel Monte Carlo method for jump-diffusion SDEs}
\author{Yuan~Xia\thanks{%
yuan.xia@oxford-man.ox.ac.uk\newline
\hspace*{1em} Xia is grateful to the China Scholarship Council for financial
support, the research  and has been
supported by the Oxford-Man Institute of Quantitative Finance. The author thanks Prof Mike Giles for the guidance as supervisor and Dr Christoph Reisinger and Dr Lajos Gergely Gyurko for their reading and helpful comments. This work is based on a talk in MCQMC10 (mcqmc.mimuw.edu.pl/Presentations/xia.pdf). All errors are mine.} \\
\multicolumn{1}{p{.7\textwidth}}{\centering\small{\emph{Oxford-Man Institute
of Quantitative Finance and Mathematical Institute, Oxford, U.K.}}} }
\maketitle

\begin{abstract}


We investigate the extension of the multilevel Monte Carlo path simulation
method to jump-diffusion SDEs. We consider models
with finite rate activity , using a jump-adapted
discretisation in which the jump times are computed and added to the
standard uniform discretisation times. The key component in multilevel
analysis is the calculation of an expected payoff difference between a
coarse path simulation and a fine path simulation with twice as many
timesteps. If the Poisson jump rate is constant, the jump times are the same
on both paths and the multilevel extension is relatively straightforward,
but the implementation is more complex in the case of state-dependent jump
rates for which the jump times naturally differ.
\end{abstract}

\date{}




\section{Introduction}

In the Black-Scholes Model, the price of an option is given by the expected
value of a payoff depending upon the solution of a stochastic differential
equation(SDE) satisfied by the stock price. The model assumes that the
behavior of the stock price is depicted by a SDE driven by Brownian motion,
\begin{equation}
\D S(t)=a(S,t)\,\D t+b(S,t)\,\D W(t),\quad 0\leq t\leq T,  \label{eq:SDE}
\end{equation}%
with given initial data $S_{0}$.

Although this model is widely used, the fact that asset returns are not
log-normal has motivated people to suggest models which better capture the
characteristics of the stock price dynamics. Merton\cite{merton76} proposed
a jump-diffusion process for the stock price. To be specific, the stock
price follows a jump-diffusion SDE:

\begin{equation}
\D S(t)=a(S(t-),t)\D t+b(S(t-),t)\D W(t)+c(S(t-),t)\D J(t),\quad 0\leq t\leq
T,  \label{eq:jumpSDE}
\end{equation}%
where the jump term $J(t)$ is a compound Poisson process $%
\sum_{i=1}^{N(t)}(Y_{i}-1)$, the jump magnitude ~$Y_{i}$ has a prescribed
distribution, and $N(t)$ is a Poisson process with intensity $\lambda $,
independent of the Brownian motion. Due to the existence of jumps, the
process is a c\`{a}dl\`{a}g process, i.e. having right continuity with left
limits. We note that $S(t-)$ denotes the left limit of the process while $%
S(t)=\lim_{s\rightarrow t+}S(t)$. In \cite{merton76}, Merton also assumed
that $\log Y_{i}$ has a normal distribution with mean $a$ and variance $b$,
namely $\log Y_{i}\sim N(a,b)$.

There are several ways to generalize Merton model from different aspects. A
possible way is to consider the case where the frequency of jump is
infinite, where general L\'{e}vy processes can be used. Another direction
(for example in \cite{gm03}) is to introduce dependency between parameters,
leaving the expected number of jumps finite within finite horizon. As a
particular numerical examples, we consider the case where instantaneous jump
rate relies on the stock price, namely $\lambda(t)=\lambda(S_t,t)$.

In pursuit of risk-neutral pricing of options, we are interested in the
expected value of a function of the terminal state, $f(S(T))$. In the simple
case of a European option, the expectation can be directly simulated, while
in the case of Asian, lookback and barrier options the valuation depends on
the entire path $S(t),0\leq t\leq T$. The expected value can be estimated by
a simple Monte Carlo method with a proper numerical discretisation scheme.
However, to achieve a root-mean-square (RMS) error of $O(\eps)$ using an
Euler-Maruyama discretisation would require $O(\eps^{-2})$ independent
paths, each with $O(\eps^{-1})$ timesteps, leading to a computational
complexity of $O(\eps^{-3})$. This is quite time consuming compared to the
case of path-independent options.

Giles \cite{giles07b} \cite{giles08b} has recently introduced a multilevel
Monte Carlo path simulation method for the option pricing calculation. This
improves the computational efficiency of Monte Carlo path simulation by
combining results using different numbers of timesteps. This can be viewed
as a generalisation of the two-level method of Kebaier \cite{kebaier05} and
is also similar in approach to Heinrich's multilevel method for parametric
integration \cite{heinrich01}. The first paper \cite{giles08b} proposed the
multilevel Monte Carlo approach and proved that it can lower the
computational complexity of path-dependent Monte Carlo evaluations to $O(\eps%
^{-2}(\log \eps)^{2})$, verified by numerical results using the simple
Euler-Maruyama discretisation. The second paper \cite{giles07b} demonstrated
that the computational cost can be further reduced to $O(\eps^{-2})$ by
using the Milstein discretisation. This has been extended by Dereich and
Heidenreich \cite{dh10,dereich11} to approximation methods for both finite
and infinite activity L\'{e}vy-driven SDEs with globally Lipschitz payoffs.
The work in this paper differs in considering simpler finite activity
jump-diffusion models, but more challenging non-Lipschitz payoffs, and also
uses a more accurate Milstein discretisation to achieve an improved order of
convergence for the multilevel correction variance which will be defined
later.

In this paper we applies the multilevel approach to the Monte Carlo
simulation of path-dependent option pricing under jump-diffusion processes.
We first consider the case where the jump rate is constant then take into
account the state-dependent rate case. In both cases, in order to calculate
coarse-path samples from fine-path sample using brownian interpolation , we
adopt a jump-adapted Milstein discretisation scheme proposed by \cite%
{Platen82}, which explicitly simulates the times when jumps occur.
Furthermore, we construct multilevel estimators for corresponding
path-dependent payoffs coping with challenges caused by jumps. Through
constructing payoff estimators by Brownian bridge technique, high order
multilevel correction term variance convergence rate is achieved. In the
state-dependent rate case, we use two approaches, which are called
cumulative intensity method and thinning method to tackle the
unsynchronization of jump times in the fine and coarse grids. Numerical
results show similar improvement in computational efficiency compared with
previous achievement \cite{giles07b} for diffusion processes. Generally,
using the jump-adapted Milstein scheme with the multilevel approach, we can
reduce the computation cost to $O(\eps^{-2})$ in terms of RMS error $\eps$.

In the following parts of the paper, we first review the Multilevel Monte
Carlo method for diffusion processes. The next section describes the
jump-adapted discretisation of jump-diffusion processes and its advantages
for facilitating the multilevel approach. Then we discuss the path
simulation and estimator construction for the jump-adapted discretisation
with the multilevel approach and present numerical results of Asian,
lookback, barrier and digital options. The next part establishes two methods
to deal with state-dependent intensity. The final section draws conclusions
and indicates directions of future research.

\section{Multilevel Monte Carlo method}

Suppose we perform Monte Carlo path simulations with fixed grid timesteps $%
h_{\ell} = 2^{-\ell}\, T$, $l=0, 1, \ldots, L$. For a given Brownian path $W(t)$,
let $P$ denote the payoff, and let $\hP_{\ell}$ denote its approximation by a
numerical scheme with timestep $h_{\ell}$. As a result of the linearity of the
expectation operator,
\begin{equation}
\EE[\hP_L] = \EE[\hP_0] + \sum_{\ell=1}^L \EE[\hP_{\ell} \!-\! \hP_{\ell-1}].
\label{eq:identity}
\end{equation}

Let $\hY_0$ denote an estimator for $\EE[\hP_0]$ using $N_0$ paths. Suppose
for different $l>0$, we use $N_{\ell}$ independent paths to estimate $%
\EE[\hP_{\ell}
\!-\! \hP_{\ell-1}]$.
\begin{equation}
\hY_{\ell} = N_{\ell}^{-1} \sum_{i=1}^{N_{\ell}} \left( \hP_{\ell}^{(i)} \!-\! \hP_{\ell-1}^{(i)}
\right).  \label{eq:est_l}
\end{equation}
The Multilevel method facilitates the fact that $\VV[\hP_{\ell} \!-\! \hP_{\ell-1}]$
decreases with $l$ to adaptively choose $N_{\ell}$ and hence reduce the
computational cost.

The cost reduction effect is summarized in the following theorem:

\begin{theorem}
\label{thm:cc}
Let $P$ denote a functional of the solution of stochastic differential
equation (\ref{eq:SDE}) for a given Brownian path $W(t)$, and let $\hP_{\ell}$
denote the corresponding approximation using a numerical discretisation with
timestep $h_{\ell} = 2^{-l}\, T$.

If there exist independent estimators $\hY_{\ell}$ based on $N_{\ell}$ Monte Carlo
samples, and positive constants $\alpha\!\geq\!{\ \textstyle \frac{1}{2} },
\beta, c_1, c_2, c_3$ such that

\begin{itemize}
\item[i)] $\displaystyle
\left| \EE[\hP_{\ell} - P] \right|\leq c_1\, h_{\ell}^\alpha $

\item[ii)] $\displaystyle
\EE[\hY_{\ell}] = \left\{
\begin{array}{ll}
\EE[\hP_0], & l=0 \\[0.1in]
\EE[\hP_{\ell} - \hP_{\ell-1}], & l>0%
\end{array}%
\right. $

\item[iii)] $\displaystyle
\VV[\hY_{\ell}] \leq c_2\, N_{\ell}^{-1} h_{\ell}^\beta $

\item[iv)] $C_{\ell}$, the computational complexity of $\hY_{\ell}$, is bounded by
\begin{equation*}
C_{\ell} \leq c_3\, N_{\ell}\, h_{\ell}^{-1},
\end{equation*}
\end{itemize}

then there exists a positive constant $c_4$ such that for any $\eps \!<\!
e^{-1}$ there are values $L$ and $N_{\ell}$ for which the multilevel estimator
\begin{equation*}
\hY = \sum_{\ell=0}^L \hY_{\ell},
\end{equation*}
has a mean-square-error with bound
\begin{equation*}
MSE \equiv \EE\left[ \left(\hY - \EE[P]\right)^2\right] < \eps^2
\end{equation*}
with a computational complexity $C$ with bound
\begin{equation*}
C \leq \left\{%
\begin{array}{ll}
c_4\, \eps^{-2} , & \beta>1, \\[0.1in]
c_4\, \eps^{-2} (\log \eps)^2, & \beta=1, \\[0.1in]
c_4\, \eps^{-2-(1-\beta)/\alpha}, & 0<\beta<1.%
\end{array}%
\right.
\end{equation*}
\end{theorem}

\begin{proof}
See \cite{giles08b}.
\end{proof}

In the case of the jump-adapted discretisation, $h_{\ell}$ should be taken to be
the uniform timestep at level $l$, to which the jump times are added to form
the set of discretisation times. We have to define the computational
complexity as the expected computational cost since different paths may have
different numbers of jumps. However, the expected number of jumps is finite
and therefore the cost bound in assumption iv) will still remain valid for
an appropriate choice of the constant $c_3$.

According to the theorem, the larger the variance convergence rate $\beta$ ,
the greater the reduction is the computation cost by the multilevel
algorithm. In the case of a Lipschitz continuous European payoff, using the
Milstein discretisation immediately leads to the result that $V_{\ell} = O(h_{\ell}^2)$%
, corresponding to $\beta\!=\!2$. Thus the main task to improve the
performance of the multilevel method is to use using schemes with high order
strong convergence rate and constructing appropriate estimators so that $%
\beta>1$ could be achieved. For the jump-diffusion process, the objective is
to obtain $\beta>1$ through adopting a high order scheme.

\section{A Jump-adapted Milstein discretisation}

To simulate jump-diffusion processes, it is possible to use fixed time grid
schemes as for geometric Brownian motion. The Euler-Maruyama scheme for
jump-diffusion processes has $O(\sqrt{h})$ strong convergence (\cite{pb10}).
However, it would be more difficult to pursue higher order strong
convergence. To achieve a higher order strong convergence for jump-diffusion
processes, the It\^{o}-Taylor expansion will involve some double integrals
of white noise and the Poisson random measure \cite{BP2005}, which increases
the complexity of the simulation.

Another problem which might be encountered for fixed-time grid schemes is
the construction of estimators for the payoff function of path-dependent
options. Adoption of the previous Brownian bridge interpolation is difficult
since the minimum or other functional of paths is difficult to calculate
since the joint density of diffusion and jump is much more complex than pure
diffusion one.

In order to avoid simulating double stochastic integrals as well as to
identify the time at which the jump occurs, we use the so-called
jump-adapted approximation proposed by Platen in \cite{Platen82}. This
jump-adapted scheme would improve the computational tractability compared to
other fixed time grid discretisation schemes with the same weak/strong
convergence order.

Suppose that we have simulated the jump time grid $\mathbb{J}%
=\{\tau_1,\tau_2,\ldots,\tau_m\}$, which includes times at which jumps occur
in the $[0,T]$. On the other hand, consider a fixed time grid constituted by
$N$ timesteps, $t_i^{\prime }=i\times \frac{T}{N},~i=1,\ldots,N,$ which is
used in discretisation schemes of Brownian SDEs. Now consider a
superposition of them as a new grid $\mathbb{T}=\{0=t_0<t_1<t_2<\ldots<t_M=T%
\}$. As a result, the length of timestep of the new grid will be no greater
than $h=\frac{T}{N}$.

Within every timestep of the new grid, the diffusion part is separated from
the effect of the possible jump, because the jump only occurs at the grid
point. Thus we can approximate the path with established schemes for
diffusion processes, and deal with corresponding adjustment if there is a
jump at the right end of interval. This procedure is called the jump-adapted
scheme.


The algorithm of simulation via the jump-adapted scheme could be described
as the following steps:

\begin{enumerate}
\item Set $i=1,~j=1$, $t^{\prime }=t^{\prime }_0$;

\item Generate jump time $\tau_i$ in terms of its distribution;

\item While $(\tau_i<t^{\prime }_j)$ do

(1). Simulate the process within $[t^{\prime },\tau_i)$, in which the
process is driven purely by Brownian motion, then simulate the jump at $%
\tau_i$;

(2). Set the timestep $h_i=\tau_i-t^{\prime }$, $t^{\prime }=\tau_i$;

(3). $i=i+1$, and generate next jump time $\tau_i$ in terms of its
distribution;

\item Simulate the process within $[t^{\prime },t^{\prime }_j]$;

\item Set the timestep $h_i=\tau_i-t^{\prime }$, $j=j+1$, $t^{\prime
}=t^{\prime }_j$ and goto 3. 
\end{enumerate}

Now we introduce a jump-adapted Milstein scheme for a scalar jump-diffusion
SDE

\begin{eqnarray*}
\hS_{n+1}^{-} &=&\hS_{n}+a_{n}\,h_{n}+b_{n}\,\Delta W_{n}+{\textstyle\frac{1%
}{2}}\,\frac{\partial b_{n}}{\partial S}\,b_{n}\,(\Delta W_{n}^{2}-h_{n}), \\
\hS_{n+1} &=&\left\{
\begin{array}{lll}
& \hS_{n+1}^{-}+c(\hS_{n+1}^{-},t_{n+1})(Y_{i}-1), & \text{ when}%
~t_{n+1}=\tau _{i}; \\
& \hS_{n+1}^{-}, & \text{ otherwise}.%
\end{array}%
\right.
\end{eqnarray*}%
Where the subscript $n$ is used to denotes the timestep index, $\hS_{n}^{-}=%
\hS(t_{n}-)$ is the left limit of the approximated solution, and $Y_{i}$ is
the jump magnitude at $\tau _{i}$.

In sum, jump-adapted schemes explicitly compute jump times, which are
relatively rare in the entire time span. Thus, compared to fixed time
schemes, they save the computation cost for generating Poisson random
numbers when the timestep tends to zero. Furthermore, as we will see later,
in terms of path simulation, jump-adapted discretisations have a very
crucial property that keeps the multilevel approach valid: within each
timestep we can neglect the jump term and only take the Brownian component
into consideration. As a matter of fact, the scheme can conveniently adopt
the Brownian bridge technique used for estimator construction in the
previous paper \cite{giles07b} so that improved convergence could be
obtained as well.


\section{Multilevel approach in the presence of jump}

In all of the cases to be presented, we simulate the paths using the
jump-adapted Milstein scheme proposed above.

In the case of a jump-diffusion process, since theorem of computation
complexity in the case of diffusion processes requires the weak convergence
and ML estimator convergence of discretisation schemes, we have to justify
them accordingly for different construction of estimators. These numerical
analysis is being done in a working paper in preparation \cite{xg11b}.

Apart from theoretical issues, there arise two challenges in the
implementation. The first problem is that the path simulation on coarse
levels needs to be revised in the presence of varying timesteps of the
jump-adapted time grid. Another is how to devise suitable estimators for
various payoffs in coping with a jump-adapted time grid. The third concern
is whether optimal samples on each level used by previous algorithm in \cite%
{giles08b} should be modified. Due to presence of jump, the computational
cost is
\begin{equation*}
Cost=\sum_{\ell=0}^{L}\sum_{i=1}^{N_{\ell}}(N_{T}^{(i)}+2^{\ell}),
\end{equation*}%
where $N_{T}^{(i)}$ is the number of jumps in each scenario. However, the
expected number of jumps is finite and therefore the cost bound in
assumption iv) will still remain valid for an appropriate choice of the
constant $c_{3}$ therefore using previous algorithm should work well. In
implementation numerical results indicates that this is appropierate for
small $\lambda.$


\subsection{Path simulation of multilevel approach in the jump-adapted time
grid}

When the multilevel approach is applied to the simulation with a fixed time
discretisation of a jump-diffusion process, the algorithm maintains the
orinigal framework straightforward. While for jump-adapted schemes,
construction of coarse path simulation needed for the estimators (\ref%
{eq:est_l}) of the payoff differs from previous case.

In the case of fixed time grid for geometric Brownian motion, every path
sample used for calculating $\hP_{\ell}^{(i)} \!-\! \hP_{\ell-1}^{(i)}$ comes from
two discrete approximations of Brownian path, which are called the fine path
and the coarse path. For every $l$, every timestep of coarse grid is
completely the same as the corresponding two timesteps of the fine grid
starting from the same endpoint. Since the brownian increment of the coarse
timestep is equidistributed to the sum of two increments of corresponding
fine ones, we can do the coarse path simulation without generating extra
random normal number.


While in the case of jump-adapted time grid, due to the presence of jump
time, the construction of a path sample in the coarse grid using the
brownian increments generated in the fine grid needs to be clarified. In
this case, the path sample is discontinuous in its jump time. To avoid such
discontinuity, notice that within each timestep of jump-adapted time grid,
the path is purely driven by Brownian component and therefore reserves
continuity. For the coarse grid and the fine grid in the same level, call
the finer grid points midpoints for short. For a particular midpoint, the
timestep, which is formed by the last jump time before this midpoint and the
first jump time after it is called midpoint timestep. In this timestep,
Brownian increment of the coarse path sample can be obtained by the
summation of Brownian increments of corresponding two timestep in fine grid.
In remaining timesteps of coarse grid, construction of the coarse path
sample uses the same Brownian increment as the one in fine grid does. Hence
we have defined the coarse path sample construction according to this
midpoint construction, which can be seen clearly in figure \ref{midpoint}.
\begin{figure}[t!]
\begin{center}
\hspace{-1in} \includegraphics[width=1.2\textwidth]{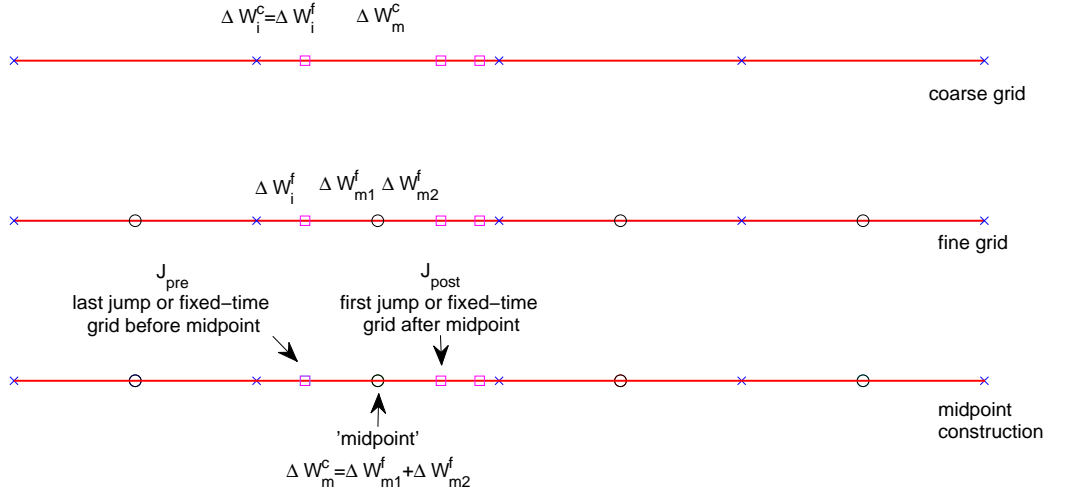}
\end{center}
\par
\vspace{-.1in}
\caption{Midpoint construction}
\label{midpoint}
\end{figure}



%

\subsection{Estimator construction}


For the estimator construction, there comes concern whether extra bias is
introduced for estimator in the jump-diffusion process. To secure the
correctness of the identity\ref{eq:identity}, we must avoid introducing any
undesired bias, hence it is required that
\begin{equation}
\EE[\hP_{\ell}^f] = \EE[\hP_{\ell}^c].  \label{eq:equality}
\end{equation}

This means that the definitions of $\hP_{\ell}$ when estimating $E[\hP_{\ell} \!-\! \hP%
_{\ell-1}]$ and $E[\hP_{\ell+1} \!-\! \hP_{\ell}]$ must have the same expectation.

In the case of path-dependent payoffs in geometric Brownian motion, we
approximate the payoff function by Brownian bridge interpolation technique
using path values on the fine and the coarse grid. This technique is also
available in jump case to reduce the variance of the estimator, where the
coarse-path estimator will involve the information from generations of
fine-path sample in the two corresponding timesteps. One thing to notice in
jump-adapted schemes is that we can utilise this construction only for the
timesteps including midpoint. For other timesteps of coarse grid, the
construction of coarse-grid estimators will be the same as the fine-grid
one. Estimators will be discussed in the following respectively
corresponding to their payoffs.


In the following we will show the numerical results of several options. All
of them are done for Merton model in which the jump-diffusion SDE under
risk-neutral measure is
\begin{equation*}
\dfrac{\D S(t)}{S(t-)}=(r-\lambda m)\,\D t+\sig\,\D W(t)+\D J(t),\quad 0\leq
t\leq T.
\end{equation*}%
where $\lambda $ is the jump intensity and jump magnitude satisfies $\log
Y_{i}\sim N(a,b)$, $r$ is the risk-free rate, $\sig$ is the volatility of
stock price and $m=\EE[Y_i]-1$ is the compensator to ensure the discounted
stock price is a martingale. All of the simulations in this section use the
parameter values $S_{0}\!=\!100$, $K\!=\!100$, $T\!=\!1$, $r\!=\!0.05$, $\sig%
\!=\!0.2$, $a\!=\!0.1$, $b\!=\!0.2$, $\lambda \!=\!1$. We thank Giles
providing the code for \cite{giles08}, based on which we can produce the current code
to generate numerical results and figures.

\subsection{Vanilla call option}

\begin{figure}[t]
\begin{center}
\includegraphics[width=\textwidth]{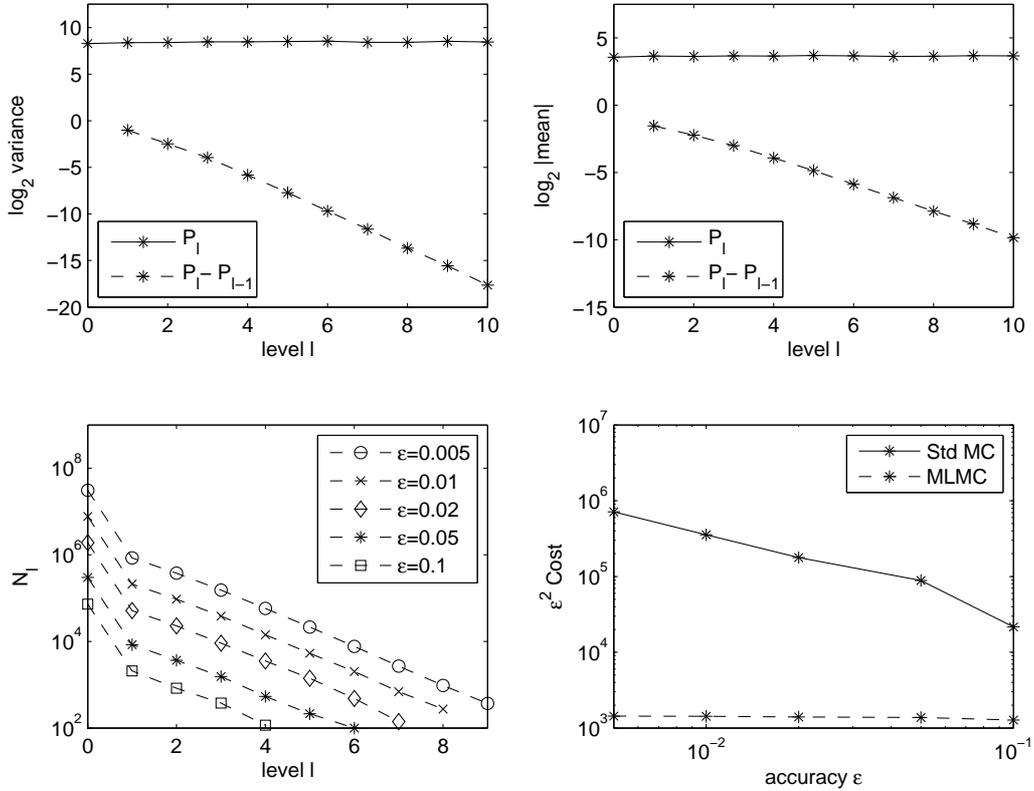}
\end{center}
\par
\vspace{-0.2in}
\caption{Vanilla option}
\label{fig:geometric Brownian motion1}
\end{figure}
For the vanilla option with the payoff $\exp (-rT)\max (S(T)-K,0)$, Figure %
\ref{fig:geometric Brownian motion1} shows the numerical results.

The top left plot shows the behaviour of the variance of both $\hP_{\ell }$
and the multilevel correction $\hP_{\ell }\!-\!\hP_{\ell -1}$, estimated
using $10^{5}$ samples so that the Monte Carlo sampling error is negligible.
The slope of the MLMC line indicates that $V_{\ell }\!\equiv \!%
\VV[\hP_\ell\!-\!\hP_{\ell-1}]\!=\!O(h_{\ell }^{2})$, corresponding to $%
\beta =2$ in condition $iii)$ of Theorem \ref{thm:cc}. The top right plot
shows that $\EE[\hP_\ell\!-\!\hP_{\ell-1}]$ is approximately $O(h_{\ell })$,
corresponding to $\alpha =1$ in condition $i)$. Noting that the payoff is
Lipschitz, both of these are consistent with the first order strong
convergence proved in \cite{pb10}.

The bottom two plots correspond to five different multilevel calculations
with different user-specified accuracies to be achieved. These use the
numerical algorithm given in \cite{giles08b} to determine the number of grid
levels, and the optimal number of samples on each level, which are required
to achieve the desired accuracy. We use the computational cost $%
\sum_{\ell=0}^{\ell}\sum_{i=1}^{N_{\ell}}(N_{T}^{(i)}+2^{\ell})$ to take account into
the effect of jump. The left plot shows that in each case many more samples
are used on level 0 than on any other level, with very few samples used on
the finest level of resolution. The right plot shows that the the multilevel
cost is approximately proportional to $\eps^{-2}$, which agrees with the
computational complexity bound in Theorem \ref{thm:cc} for the $\beta \!>\!1$
case.

\subsection{Asian option}

The payoff of the Asian option we consider is
\begin{equation*}
P = \exp(-rT)\ \max\left(0, \overline{S} \!-\! K \right),
\end{equation*}
where
\begin{equation*}
\overline{S} = T^{-1} \int_0^T S(t) \ \D t.
\end{equation*}
$n_T\!=\!T/h$ is the number of timesteps. \cite{giles07b} shows that
accuracy can be achieved by approximating the behaviour of a process within
a timestep as an I\^{t}o process with constant drift and volatility,
conditional on the computed endpoint values $\hS_n$. Taking $b_n$ to be the
constant volatility within the interval $[t_n,t_{n+1}]$, in other words, we
define brownian interpolation in the coarse grid at $t$ as

\begin{equation}
\hS(t) = \hS_n + \mu(\hS_{n+1}^- - \hS_n)+ b_n [W_t-W_n-\mu(W_{n+1}-W_n)] ,
\label{eq:midpoint}
\end{equation}
where $\mu=(t-t_n)/h,~h=t_{n+1}-t_n$.

This implies
\begin{equation*}
\int_{t_{n}}^{t_{n+1}}\hS(t)\ \D t={\ \textstyle\frac{1}{2}}%
h(S(t_{n})+S(t_{n+1}-))+b_{n}\Delta I_{n},
\end{equation*}%
where $\Delta I_{n}$ is
\begin{equation*}
\Delta I_{n}:=\int_{t_{n}}^{t_{n+1}}(W(t)-W(t_{n}))\ \D t\ -\ {\ \textstyle%
\frac{1}{2}}\,h\Delta W,
\end{equation*}%
satisfying $\Delta I_{n}\sim N(0,h^{3}/12)$ , and is independent of $\Delta
W $. Let $b_{n}=b(\hS_{n},t_{n})$, the fine-path approximated payoff would
be
\begin{equation*}
\overline{S}^{f}=T^{-1}\sum_{n=0}^{n_{T}-1}\left( {\ \textstyle\frac{1}{2}}%
\,h\,(\hS_{n}\!+\!\hS_{n+1}^{-})+b_{n}\Delta I_{n}^{f}\right).
\end{equation*}

In a jump-adapted grid, the coarse-path approximation is the same in most
timesteps except in the midpoint timestep $\Delta I_n^c$ is derived from the
fine-path values, namely 

\begin{figure}[t!]
\begin{center}
\includegraphics[width=\textwidth]{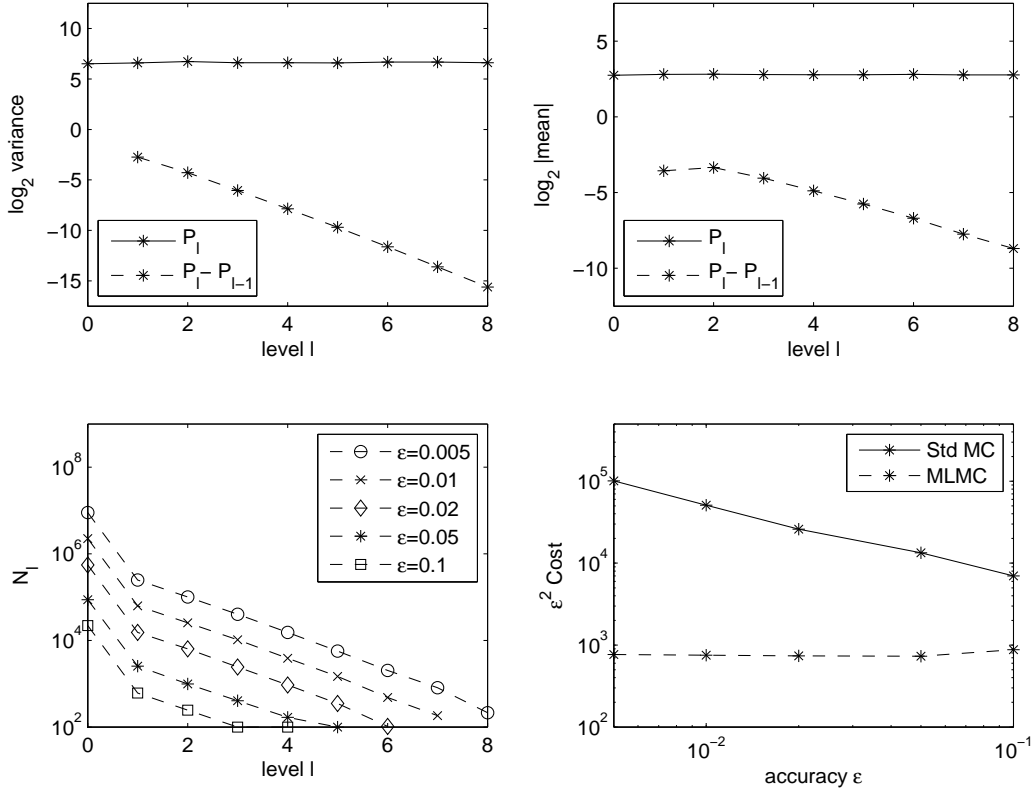}
\end{center}
\par
\vspace{-0.2in}
\caption{Asian option}
\label{fig:jd2}
\end{figure}

\begin{eqnarray*}
\lefteqn{\int_{t_n}^{t_{n+2}} (W(t)-W(t_n)) \ \D t\ -\ \textstyle \frac{1}{2}%
(t_{n+2}-t_{n}) (W(t_{n+2})-W(t_n))} && \\
&=& \int_{t_n}^{t_{n+1}} (W(t)-W(t_n)) \ \D t\ -\ {\ \textstyle \frac{1}{2} }%
\, (t_{n+1}-t_{n}) \left(W(t_{n+1})-W(t_n)\right) \\
&+& \int_{t_n+h}^{t_{n+2}} (W(t)-W(t_{n+1})) \ \D t\ -\ {\ \textstyle \frac{1%
}{2} }\, (t_{n+2}-t_{n+1}) \left(W(t_{n+2})-W(t_{n+1})\right) \\
&+& {\ \textstyle \frac{1}{2} }\,
(t_{n+2}-t_{n+1})\left(W(t_{n+1})-W(t_n)\right) - {\ \textstyle \frac{1}{2} }%
\, (t_{n+1}-t_{n})\left(W(t_{n+2})-W(t_{n+1})\right),
\end{eqnarray*}
and thus
\begin{equation*}
\Delta I^c = \Delta I^{f1}_n + \Delta I^{f2}_n + {\ \textstyle \frac{1}{2} }%
\, (t_{n+2}-t_{n}) (\mu\Delta W^{f1} - (1-\mu)\Delta W^{f2}),
\end{equation*}
where $\mu_n=(t_{n+2}-t_{n+1})/(t_{n+2}-t_{n})$, $\Delta I^c$ is the value
for the coarse timestep in midpoint timestep; $\Delta I^{f1}$ and $\Delta
W^{f1}$ are the values for fine timestep in the first fine-path timestep
constituting the midpoint timestep; $\Delta I^{f2}$ and $\Delta W^{f2}$ are
the values for the fine timestep in the latter one constituting the midpoint
timestep.

Figure \ref{fig:jd2} shows the numerical results for parameters $%
S(0)\!=\!100 $, $K\!=\!100$, $T\!=\!1$, $r\!=\!0.05$, $\sig\!=\!0.2$, $%
a\!=\!0.1$, $b\!=\!0.2$, $\lambda\!=\!1$. All the results are similar to the
pure diffusion case.

\subsection{Lookback option}

The payoff of the lookback option we consider is
\begin{equation*}
P=\exp (-rT)\left( S(T)-\min_{0\leq t\leq T}S(t)\right) .
\end{equation*}%
Previous work \cite{giles07b} achieved a second order convergence rate for
the multilevel correction variance using the Milstein discretisation and an
estimator constructed by approximating the behaviour within a timestep as an
It\^{o} process with constant drift and volatility, conditional on the
endpoint values $\hS_{n}$ and $\hS_{n+1}$. Brownian Bridge results (see
section 6.4 in \cite{glasserman04}) give the minimum value within the
timestep $[t_{n},t_{n+1}]$, conditional on the end values, as
\begin{equation}
\hS_{n,min}={\ \textstyle\frac{1}{2}}\left( \hS_{n}+\hS_{n+1}-\sqrt{\left( %
\hS_{n+1}\!-\!\hS_{n}\right) ^{2}-2\,b_{n}^{2}\,h\log U_{n}}\ \right) ,
\label{eq:lookback1}
\end{equation}%
where $b_{n}$ is the constant volatility and $U_{n}$ is a uniform random
variable on $[0,1]$. The same treatment can be used for the jump-adapted
discretisation in this paper, except that $\hS_{n+1}^{-}$ must be used in
place of $\hS_{n+1}$ in \eqref{eq:lookback1}.

Equation \eqref{eq:lookback1} is used for the fine path approximation, but a
different treatment is used for the coarse path, as in \cite{giles07b}. This
involves a change to the original telescoping sum in \eqref{eq:identity}
which now becomes
\begin{equation}
\EE[\hP^f_L] = \EE[\hP^f_0] + \sum_{\ell=1}^L \EE[\hP^f_\ell \!-\!
\hP^c_{\ell-1}],  \label{eq:identity2}
\end{equation}
where $\hP^f_\ell$ is the approximation on level $\ell$ when it is the finer
of the two levels being considered, and $\hP^c_\ell$ is the approximation
when it is the coarser of the two. This modified telescoping sum remains
valid provided $\EE[\hP^f_\ell] = \EE[\hP^c_\ell]. $

Considering a particular timestep in the coarse path construction, we have
two possible situations. If it does not contain one of the fine path
discretisation times, and therefore corresponds exactly to one of the fine
path timesteps, then it is treated in the same way as the fine path, using
the same uniform random number $U_n$. This leads naturally to a very small
difference in the respective minima for the two paths.

\begin{figure}[t!]
\begin{center}
\includegraphics[width=\textwidth]{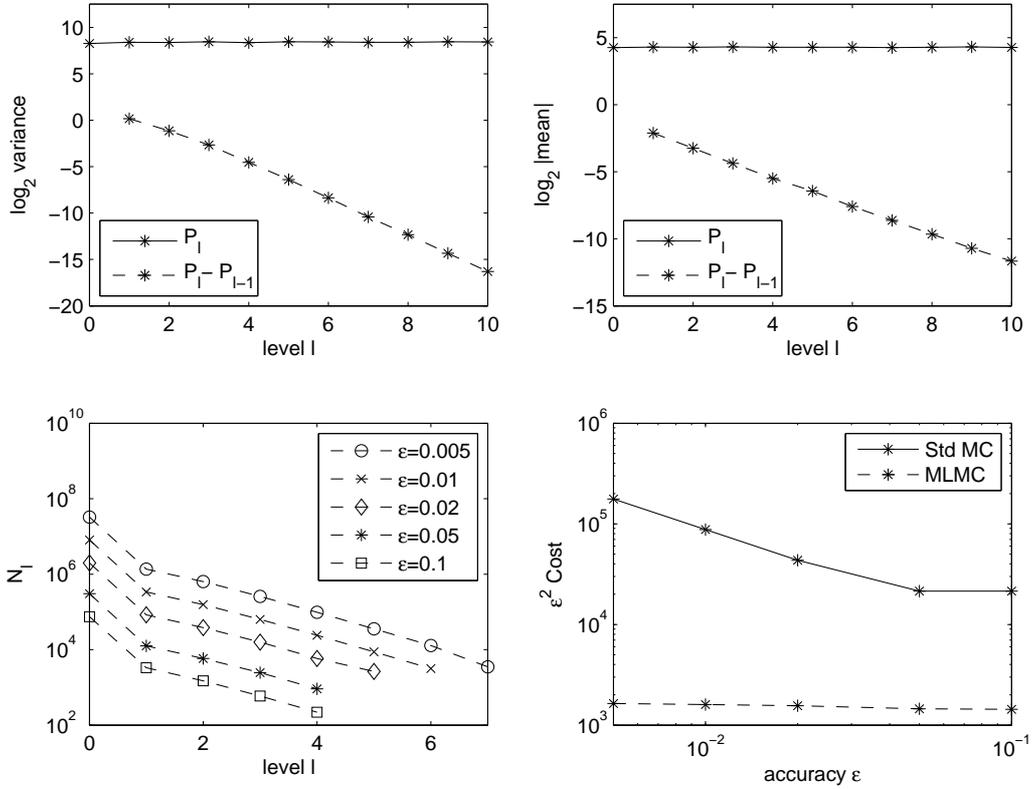}
\end{center}
\par
\vspace{-0.2in}
\caption{Lookback option}
\label{fig:geometric Brownian motion3}
\end{figure}

The more complicated case is the one in which the coarse timestep contains
one of the fine path discretisation times $t^{\prime }$, and so corresponds
to the union of two fine path timesteps. In this case, the value at time $%
t^{\prime }$ is given by the conditional Brownian interpolant
\begin{equation}
\hS(t^{\prime })=\hS_{n}+ \mu\, (\hS_{n+1}^{-}-\hS_{n}) +b_{n} \left(
W(t^{\prime })-W_{n}-\mu\, (W_{n+1}-W_{n})\right),  \label{eq:midpoint}
\end{equation}%
where $\mu =(t^{\prime }-t_{n})/(t_{n+1}-t_{n})$ and the value of $%
W(t^{\prime })$ comes from the fine path simulation. Given this value for $%
\hS(t^{\prime })$, the minimum values for $S(t)$ within the two intervals $%
[t_n,t^{\prime }]$ and $[t^{\prime },t_{n+1}]$ can be simulated in the same
way as before, using the same uniform random numbers as the two fine
timesteps.

The equality $\EE[\hP^f_\ell] = \EE[\hP^c_\ell] $ is respected in this
treatment because $W(t^{\prime })$ comes from the correct distribution,
conditional on $W_{n+1}, W_{n}$, and therefore, conditional on the values of
the Brownian path at the set of coarse discretisation points, the computed
value for the coarse path minimum has exactly the same distribution as it
would have if the fine path algorithm were applied.

Further discussion and analysis of this is given in \cite{xg11b}, including
a proof that the strong error between the analytic path and the conditional
interpolation approximation is at worst $O(h\,\log h)$.

Figure \ref{fig:geometric Brownian motion3} presents the numerical results.
The results are very similar to those obtained by Giles for geometric
Brownian motion \cite{giles07b}. The top two plots indicate second order
variance convergence rate and first order weak convergence, both of which
are consistent with the $O(h\,\log h)$ strong convergence. The computational
cost of the multilevel method is therefore proportional to $\eps^{-2}$, as
shown in the bottom right plot.

\subsection{Barrier option}

We consider a down-and-out call barrier option for which the discounted
payoff is
\begin{equation*}
P=\exp (-rT)\,(S(T)\!-\!K)^{+}\,\mathbbm{1}_{\left\{ M_{T}>B\right\} },
\end{equation*}%
where $M_{T}=\min_{0\leq t\leq T}S(t).$ The jump-adapted Milstein
discretisation with the Brownian interpolation gives the approximation
\begin{equation*}
\hP=\exp (-rT)\,(\hS(T)\!-\!K)^{+}\,\mathbbm{1}_{\left\{ \widehat{M}%
_{T}>B\right\} }
\end{equation*}%
where $\widehat{M}_{T}=\min_{0\leq t\leq T}\hS(t)$. This could be simulated
in exactly the same way as the lookback option, but in this case the payoff
is a discontinuous function of the minimum $M_{T}$ and an $O(h)$ error in
approximating $M_{T}$ would lead to an $O(h)$ variance for the multilevel
correction.

Instead, following the approach of Cont \& Tankov (see page 177 in \cite%
{ct04}), it is better to use the expected value conditional on the values of
the discrete Brownian increments and the jump times and magnitudes, all of
which may be represented collectively as $\mathcal{F}$. This yields
\begin{eqnarray*}
\lefteqn{\hspace*{-0.5in}\mathbb{E}\left[ \exp (-rT)\ (\hS(T)\!-\!K)^{+}%
\mathbbm{1}_{\left\{ \widehat{M}_{T}>B\right\} }\right] } \\
&=&\mathbb{E}\left[ \exp (-rT)\ (\hS(T)\!-\!K)^{+}\mathbb{E}\left[ %
\mathbbm{1}_{\left\{ \widehat{M}_{T}>B\right\} }\mid \mathcal{F}\right] %
\right]  \\
&=&\mathbb{E}\left[ \exp (-rT)\ (\hS(T)-K)^{+}\ \prod_{n=0}^{n_{T}-1}\hp_{n}%
\right]
\end{eqnarray*}%
where $\hp_{n}$ denotes the conditional probability that the path does not
cross the barrier $B$ during the $n^{th}$ timestep:
\begin{equation}
\hp_{n}=1-\exp \left( \frac{-2\,(\hS_{n}\!-\!B)^{+}(\hS_{n+1}^{-}\!-\!B)^{+}%
}{b_{n}^{2}\ (t_{n+1}-t_{n})}\right) .  \label{eq:barrier1}
\end{equation}

\begin{figure}[t!]
\begin{center}
\includegraphics[width=\textwidth]{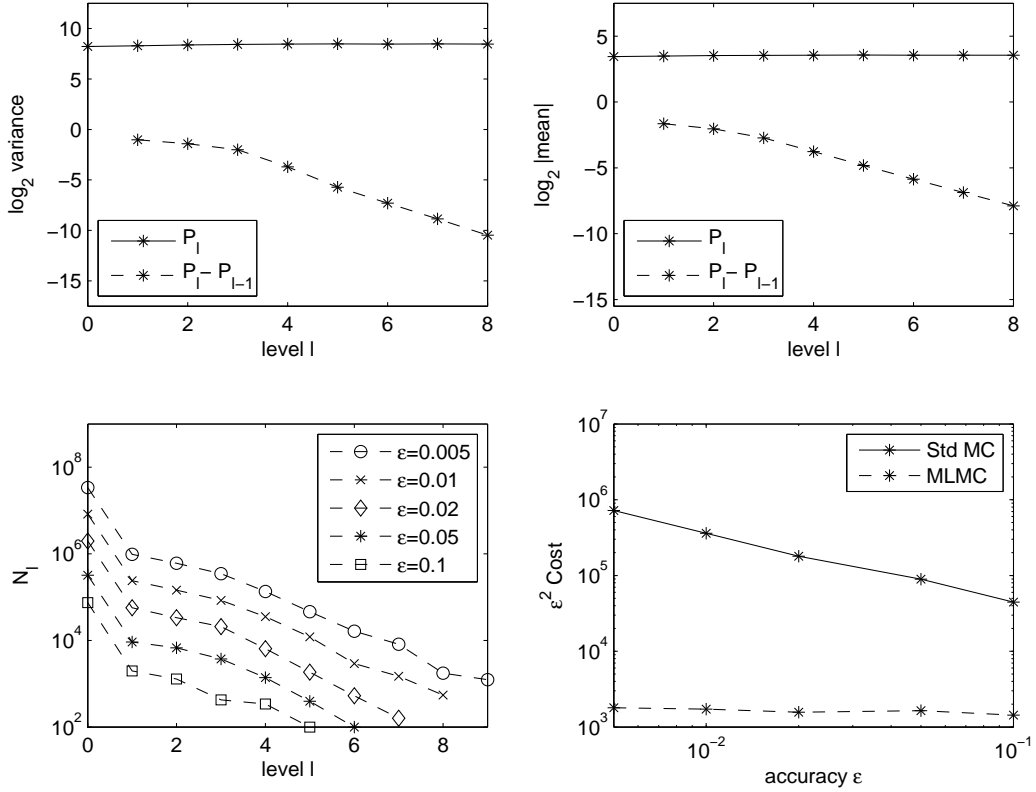}
\end{center}
\par
\vspace{-0.2in}
\caption{Barrier option}
\label{fig:geometric Brownian motion4}
\end{figure}

For fine-path value, we compute $\hp_{n}^{f}$ where $b_{n}$ is defined equal
to $b(\hS_{n},t_{n})$ within each timestep of the jump-adapted grid. Note
that $S_{n+1}$ in \eqref{eq:barrier1} should be replaced by the value of
left limit of the endpoint $S_{n+1}^{-}$, as in the lookback calculation.

For the coarse path calculation, we again deal separately with two cases.
When the coarse timestep does not include a fine path time, then we again
use (\ref{eq:barrier1}). In the other case, when it includes a fine path
time $t^{\prime }$ we evaluate the Brownian interpolant at $t^{\prime }$ and
then use the conditional expectation to obtain
\begin{eqnarray}
\hp_{n} &=&\left\{ 1-\exp \left( \frac{-2\,(\hS_{n}\!-\!B)^{+}(\hS%
(t)\!-\!B)^{+}}{b_{n}^{2}\ (t^{\prime }-t_{n})}\right) \right\}   \notag \\
&&\hspace{-0.2in}\times \ \left\{ 1-\exp \left( \frac{-2\,(\hS(t^{\prime +}(%
\hS_{n+1}^{-}\!-\!B)^{+}}{b_{n}^{2}\ (t_{n+1}-t^{\prime })}\right) \right\} .
\label{eq:barrier2}
\end{eqnarray}

Figure \ref{fig:geometric Brownian motion4} shows the numerical results for $%
K\!=\!100$, $B\!=\!85$. The top left plot shows that the multilevel variance
is $O(h_{\ell }^{\beta })$ for $\beta \approx 3/2$. This is similar to the
behavior for a diffusion process \cite{giles07b}. The bottom right plot
shows that the computational cost of the multilevel method is again almost
perfectly proportional to $\eps^{-2}$.

\subsection{Digital option}

The digital option considered here has the discounted payoff
\begin{equation*}
P = \exp(-rT)\ \mathbbm{1}_{\{S(T)>K\}}.
\end{equation*}

In \cite{giles07b}, a multilevel variance convergence rate of $O(h_{\ell}^{3/2})$
is achieved by smoothing the payoff using conditional expectation given the
brownian increments terminated one timestep before reaching the terminal
time $T$. The estimator is the probability that $\hS_{n_{\ell}} > K$ under
assumption of simple Brownian motion with constant drift $a_{n_{\ell}-1}$ and
volatility $b_{n_{\ell}-1}$ within last timestep where $n_{\ell}$ denotes number of
fine-path timesteps:

\begin{eqnarray*}
\EE[\hP_{\ell} - \hP_{\ell-1}] &=& \EE[~\EE[~f(\hS_{n_{\ell}-1}^f)-f(\hS_{n_{\ell}-1}^c)~|~
\Delta W_{i},\ i=1,\ldots,n_{\ell}-1]~] \\
&=& \EE[~\Phi \left( \frac{\hS_{n_{\ell}-1}^f \!+\! a_{n_{\ell}-1}^f h -
K}{b_{n_{\ell}-1}^f \sqrt{h}}\right) - \Phi \left( \frac{\hS_{n_{\ell}-2}^c \!+\!
2a_{n_{\ell}-2}^c h \!+\! b_{n_{\ell}-2}^c \Delta W_{n_{\ell}-1} - K}{b_{n_{\ell}-2}^c
\sqrt{h}}\right)~],
\end{eqnarray*}
where $\Phi$ is the cumulative density function of Normal variable, $h$ is
fine-path fixed-time timestep.

In the jump-adapted time grid, the relation between last jump time and the
last timestep before expiry leads to different expressions of above
conditional expectation estimator. Let In fact, there would be three cases:

\begin{enumerate}
\item The last jump time $J$ happens before penultimate fixed-time timestep,
i.e. $J<(N-2)\frac{T}{N}$, where $N$ is the number of timesteps in
fixed-time fine grid;

\item The last jump time is within the last fixed-time timestep , i.e. $%
J>(N-1)\frac{T}{N}$;

\item The last jump time is within penultimate fixed-time timestep, i.e. $%
(N-1)\frac{T}{N}>J>(N-2)\frac{T}{N}$.
\end{enumerate}

Correspondingly, different fine-path and coarse-path estimators are shown in
the following.
\begin{figure}[t]
\begin{center}
\hspace{-.2in} \includegraphics[width=1.1\textwidth]{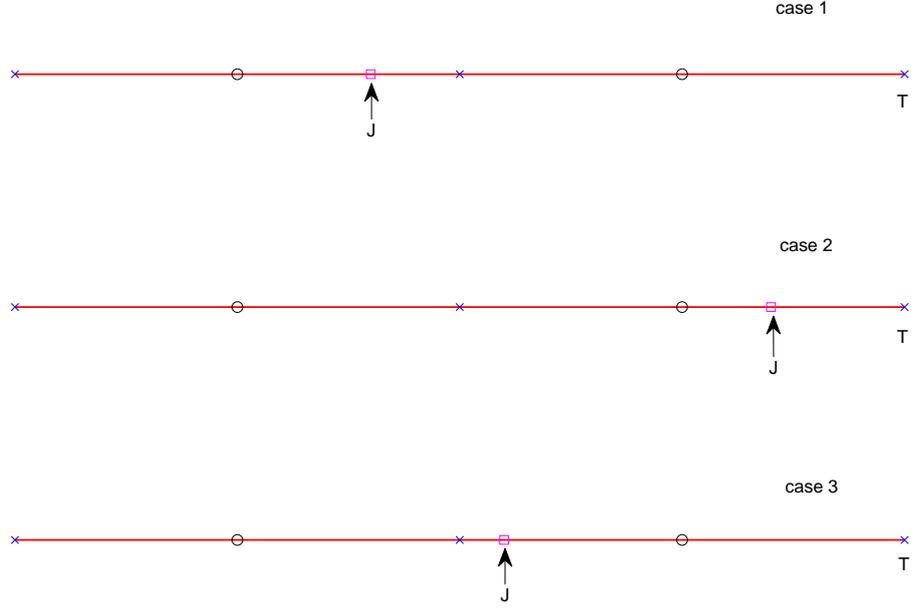}
\end{center}
\par
\vspace{-.3in}
\caption{Construction of conditional expectation estimator for Digital
option }
\label{Digital option}
\end{figure}

\begin{enumerate}
\item In case 1, the fine-path grid and coarse-path grid is the same as the
previous diffusion case, hence we do not need to change the estimator, which
is the probability that $\hS_{n_{T}}>K$ under assumption of approximation
the dynamics as a constant drift $a_{n_{T}-1}^{f}\equiv \!a(\hS%
_{n_{T}-1}^{f},T\!-h_{n_{T}})$ and volatility $b_{n_{T}-1}^{f}\!\equiv \!b(%
\hS_{n_{T}-1}^{f},T\!-\!h_{\ell})$ Brownian motion within last timestep.

\begin{equation*}
\hP_{\ell}^{f}=\Phi \left( \frac{\hS_{n_{T}-1}^{f}\!+\!a_{n_{T}-1}^{f}h-K}{%
\,\left\vert b_{n_{T}-1}^{f}\right\vert \sqrt{h}}\right) .
\end{equation*}%
where $\Phi $ is the cumulative Normal distribution, $h=T/2^{\ell}$ is
fine-path fixed-time timestep.

\begin{equation*}
\hP_{\ell-1}^{c}=\Phi \left( \frac{\hS_{n_{T}-2}^{c}\!+\!2a_{n_{T}-2}^{c}h\!+%
\!b_{n_{T}-2}^{c}\Delta W_{h}-K}{\left\vert b_{n_{T}-2}^{c}\right\vert \sqrt{%
h}}\right) .
\end{equation*}

\item In case 2, last timestep of fine path would be $h_{j}=T-t_{n_{T}-1}$.
Due to discontinuity of path before last jump, we must use the same
estimator for both fine and coarse path
\begin{equation*}
\hP_{\ell}^{f}=\Phi \left( \frac{\hS_{n_{T}-1}^{f}\!+\!a_{n_{T}-1}^{f}h_{j}-K}{%
\left\vert b_{n_{T}-1}^{f}\right\vert \sqrt{h_{j}}}\right) ,
\end{equation*}

\begin{equation*}
\hP_{\ell-1}^{c}=\Phi \left( \frac{\hS_{n_{T}-1}^{c}\!+\!a_{n_{T}-1}^{c}h_{j}-K%
}{\left\vert b_{n_{T}-1}^{c}\right\vert \sqrt{h_{j}}}\right) .
\end{equation*}

\item In the last case, $J$ denotes the last jump time, and $h_{j}=T-J-h$.
We again utilise Brownian increment generated for fine path $W_{h_{j}}$.

\begin{equation*}
\hP_{\ell}^{f}=\Phi \left( \frac{\hS_{n_{T}-1}^{f}\!+\!a_{n_{T}-1}^{f}h-K}{%
\left\vert b_{n_{T}-1}^{f}\right\vert \sqrt{h}}\right) ,
\end{equation*}

\begin{equation*}
\hP_{\ell-1}^{c}=\Phi \left( \frac{\hS_{n_{T}-2}^{c}\!+\!a_{n_{T}-2}^{c}h_{j}%
\!+\!b_{n_{T}-2}^{c}\Delta W_{h_{j}}-K}{\left\vert
b_{n_{T}-2}^{c}\right\vert \sqrt{h}}\right) .
\end{equation*}
\end{enumerate}

\begin{figure}[t!]
\begin{center}
\includegraphics[width=\textwidth]{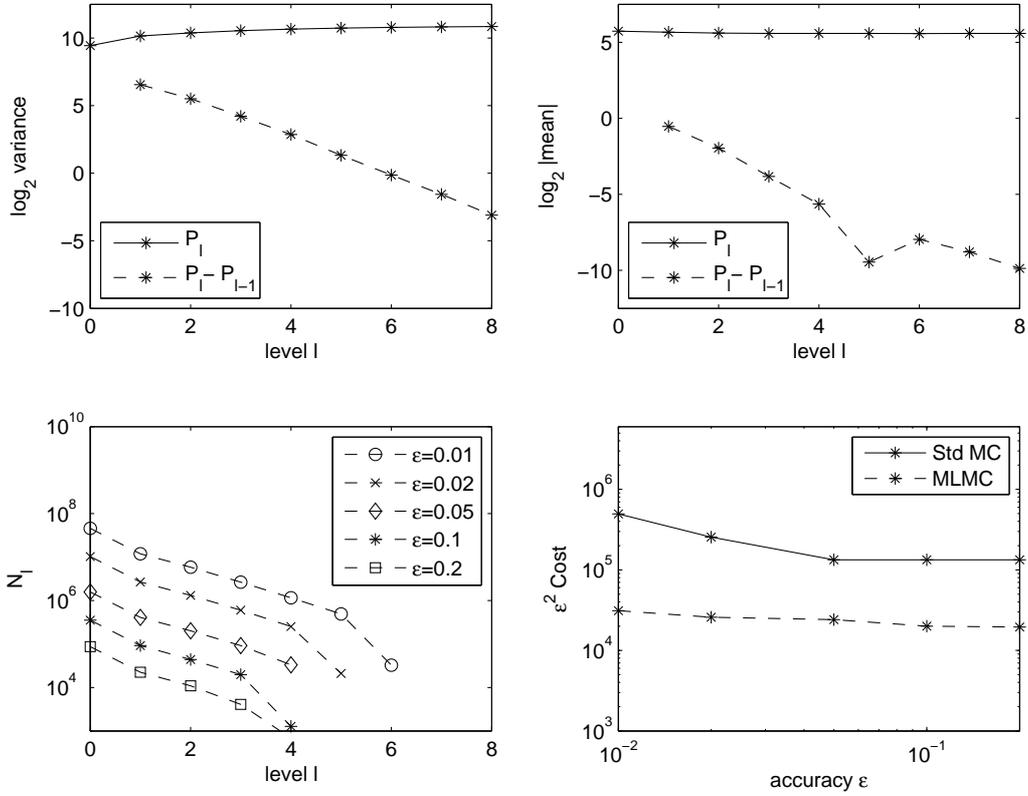}
\end{center}
\par
\vspace{-0.2in}
\caption{Digital option}
\label{fig:geometric Brownian motion5}
\end{figure}

In three cases, the conditional expectation of coarse-path estimator is
equal to fine-path one, thus equality (\ref{eq:equality}) is justified.
Figure \ref{Digital option} clearly demonstrates three cases.

Figure \ref{fig:geometric Brownian motion5} shows the numerical results for
parameters $S(0)\!=\!100$, $K\!=\!100$, $T\!=\!1$, $r\!=\!0.05$, $\sig%
\!=\!0.2$. The top left plot shows that the variance is approximately $%
O(h_{\ell}^{3/2})$, corresponding to $\beta\!=\!1.5$. The reason for this is
similar to the argument in \cite{giles07b}.

A different feature compared to the geometric Brownian motion case is that
the variance of the level 0 estimator is a constant increasing with jump
rate $\lambda$, instead of zero. The reason is simply because that the
trajectories in level 0 do vary in each simulation.

\section{Path-dependent rate cases}

In the case of path-dependent jump rates, which means the jump intensity
depends on the process, for instance $\lambda=\lambda(S_t,t)$, the
implementation of multilevel becomes difficult due to the fact that path may
jump at different time in the fine grid and coarse grid. These differences
in path sample might enlarge the difference of quantities used in the
computation of payoff between the fine grid and coarse grid, such as final
payoff in vanilla and digital option, minimum estimator by Brownian bridge
interpolation in lookback option and crossing probability in barrier option.
This leads to a increased variance of every single splitted multilevel
estimator and finally decrease the reduction of computation cost.

To tackle this obstacle, we propose two approaches. The first one
facilitates the idea of change of measure in dealing with multilevel method
for discontinuous payoffs, as we used before in dealing with digital option.
The second one uses thinning technique to simulate the desired jump time
with inhomogeneous rate through acceptance-rejection procedure. We again
need to change the measure when evaluating $P_{\ell}-P_{\ell-1}$ in both grids to
reduce variance.

\subsection{Cumulative intensity method}

In the first approach, which we call cumulative intensity method, we
generate jump times from cumulative intensity which is computed in
accordance with the dynamics evolution. We can see the idea from the case of
deterministic time-inhomogeneous rate. For a time-inhomogeneous Poisson
process whose instaneous intensity is $\lambda (s),$ the distribution of
next jump time $\tau _{i+1}$ given $\tau _{i}$ is
\begin{equation}
P(\tau _{i+1}-\tau _{i}\leq t|\tau _{i})=1-\exp (-\int_{\tau _{i}}^{\tau
_{i}+t}\lambda (s)\D s).  \label{jumpmarkov}
\end{equation}

construct jump-adapted schemes by approximating the instaneous jump rate
adaptively with the evolution of the path sample at each time grid point. In
other words, since jump intensity is path-dependent, the jump-adapted time
discretisation grid should be generated corresponding to the evolution of
the underlying process. In

In other words,
\begin{equation*}
\int_{\tau _{i}}^{\tau _{i+1}}\lambda (s)\D s\sim \mathcal{E}(1).%
\mbox{~(Exponential distribution with parameter 1.)}
\end{equation*}%
We can use this property to generate $\tau _{i+1}$:

\begin{equation*}
\Delta \tau _{i+1}=\inf \{{t\geq \tau _{i}:\int_{\tau _{i}}^{t}\lambda (s)\D %
s\>>\mathcal{E}_{i+1}\}},~~\mathcal{E}_{i}\sim \mathcal{E}(1).
\end{equation*}%
So%
\begin{equation*}
\tau _{i+1}=\tau _{i}+\Delta \tau _{i+1}.
\end{equation*}


\bigskip If the cumulative intensity
\begin{equation*}
\Lambda (t)=\int_{0}^{t}\lambda (s)\D s
\end{equation*}%
does not admit an explicit expression or it is computational intensive, we
can use simple quadrature with linear interpolation to generate the
approximated $\tau _{i+1}:$ $\ \htau_{i}$. \
\begin{equation*}
\int_{\htau_{i}}^{\htau_{i}+t_{n}}\lambda (s)\D s\approx
\sum_{i=1}^{n}\lambda _{i}h_{i}:=\Lambda _{n}
\end{equation*}%
where $t_{n}$ is the next $n$th point in the fixed-timestep grid, $h_{i}$ is
the timestep of the grid, $\lambda _{i}=\lambda (\htau_{i}+%
\sum_{j=1}^{i-1}h_{j})$ is the instaneous intensity on the left endpoint of
each timestep.

Let $N=\inf \{n:\Lambda _{n}>\mathcal{E}_{i+1}\}$, then
\begin{equation}
\Delta \htau _{i+1}=\dfrac{\mathcal{E}_{i+1}-\Lambda _{N-1}}{\lambda _{N}}
\label{jumptime}
\end{equation}

is a approximation of $\Delta \tau _{i+1}$.

In the case of random intensity where $\lambda$ only depends on current
state, we can define $\tau_{i}$ as

\begin{equation*}
\tau_{i}:=\inf{\large \{t:\int_0^t\lambda_s\D s\geq \mathcal{E}_{1}+\cdots+%
\mathcal{E}_{i} \}.}
\end{equation*}

We can use approximation of integral and \eqref{jumptime} to generate $\htau%
_{i}$.

The algorithm of jump-adapted (Milstein) scheme with cumulative intensity
would be:

\begin{verbatim}
    Algorithm (jump-adapted Milstein scheme with cumulative intensity)
\end{verbatim}

Suppose that we have a fixed time grid constituted of $N$ timesteps, $%
t_{i}^{\prime }=i\times \frac{T}{N},~i=1,\ldots ,N$.

\begin{enumerate}
\item Let $\Lambda ,t,E=0,$ $i,\ j=1.\ $Draw an exponential r.v. $\mathcal{E}%
_{j}$ with parameter 1;

\item While $t<T,$ do

\begin{enumerate}
\item $\Lambda ^{\prime }=\Lambda +\lambda (\hS(t),t)(t_{i}^{\prime }-t),$ $%
E=\mathcal{E}_{j}+E,$

\begin{enumerate}
\item  If $\Lambda ^{\prime }>\mathcal{E}_{j},$ 

Let $h=\dfrac{\mathcal{E}%
_{j}-\Lambda }{\lambda (\hS(t),t)},$ $\Lambda =E.$ 

Mark $\htau_{j}=t+h$ as a
jump time. Let $j=j+1$ and generate $\mathcal{E}_{j};$

\item Otherwise $\Lambda =\Lambda ^{\prime }.\ h=t_{i}^{\prime }-t;\ i=i+1.$
\end{enumerate}

\item (Use Milstein scheme) Simulate the evolution of the process from $t$
to $t+h,$ obtaining value of$\ \hS(t+h)$ depending on whether $t+h$ is a
jump time. 

Let $t=t+h.$
\end{enumerate}
\end{enumerate}

As a side output, after the finishing of algorithm we have got the
approximated jump times $\htau_{j}$, which combines  $t_{i}^{\prime
}=i\times \frac{T}{N},~i=1,\ldots ,N$ forming jump-adapted grid $\mathbb{T}%
=\{0=t_{0}<t_{1}<t_{2}<\ldots <t_{M}=T\}.$

The point process $\hN_{t}$ corresponding to the stopping times $\htau_{i}$
is defined by
\begin{equation}
\hN_{t}:=\sum_{i=1}^{\infty }\ii_{\{\htau_{i}\leq t\}}.  \label{point}
\end{equation}

This process is indeed a point process with piecewise constant intensity $%
\lambda (\hS(t_{k}),t_{k}),\ ~k=0,\ldots ,M-1$ in $[t_{k},t_{k+1})$.


%
%

\subsubsection{Multilevel treatment}

When it comes to multilevel approach, the problem is that the current
intensity may be different in fine and coarse grid, which leads to different
distributions of next jump time. This causes two problems. First, we can no
more use the random number generated for the path increment of fine grid to
the one of the coarse grid, which is intrinsically unacceptable for
multilevel approach. Secondly, the different final jumps may make a big
difference between the payoffs in the fine grid and in the coarse grid,
which is a similar challenge for digital option.

To handle these problems, we change the measure of Poisson rate in
calculating the expectation in the coarse grid so that the distribution of
next jump time agrees with the one in the fine grid. To do this let us first
introduce the change of measure for Poisson processes (see section 9.3 of
\cite{ct04}).

Suppose $N_{t}\sim \mbox{Poi}(\lambda _{1})$ under some probability density $%
P_{1}$, then under probability density $P_{2}$ defined by
\begin{equation}
\dfrac{\D P_{2}}{\D P_{1}}=\exp ((\lambda _{1}-\lambda _{2})t)(\dfrac{%
\lambda _{2}}{\lambda _{1}})^{N_{t}}  \label{eq:change}
\end{equation}%
we will have $N_{t}\sim \mbox{Poi}(\lambda _{2})$, where the above term is
called Radon-Nikodym derivative.

%

In the context of multilevel approach under state-dependent intensity model
with jump-adapted scheme, we want to calculate
\begin{equation*}
\mathbb{E}[\hP_{\ell}^{f}]-\mathbb{E}[\hP_{\ell-1}^{c}],
\end{equation*}%
which we can rewrite as
\begin{equation*}
\mathbb{E}[\hP_{\ell}^{f}-\hP_{\ell-1}^{f}\hR_{\ell}].
\end{equation*}

\bigskip We shall explain the meaning of this formula: instead of defining
in the jump-adapted grid formed by cumulative intensity approximated by
coarse timestep $2h_{\ell},$ $\hP_{\ell-1}$ is defined in the jump-adapted grid
formed by fine timestep $h_{\ell}$ approximation.

The Radon-Nikodym derivative $\hR_{\ell}$ is defined by

\begin{equation*}
\hR_{\ell}(T)=\dfrac{\D \hP_{c}}{\D \hP_{f}}=\exp (\Lambda _{\ell}^{f}-\Lambda
_{\ell}^{c})\prod\limits_{k}\dfrac{\lambda _{k}^{c}}{\lambda _{k}^{f}},
\end{equation*}%
in which $T$ is maturity, %
%
denoting $\lambda _{i}^{f}=\lambda (\hS_{i}^{f},t_{i}),\lambda
_{i}^{c}=\lambda (\hS_{i}^{c},t_{i}),\ \Lambda
_{\ell}^{f}=\sum_{i=1}^{n_{\ell}}\lambda _{i}^{f}h_{i}$ and $\Lambda
_{\ell}^{c}=\sum_{i=1}^{n_{\ell}}\lambda _{i}^{c}h_{i}$ are approximated
cumulative intensities upto the maturity in the fine and coarse grid, $k$ is
the index that $\htau_{k+1}^{f}\ $is the jump time, respectively.

This construction is valid since $\mathbb{E}[\hP_{\ell-1}^{c}]=\mathbb{E}[\hP%
_{\ell-1}^{f}\hR_{\ell}],$ which can be seen as a piecewise constant extension of (%
\ref{eq:change}). This can be justified by the theorem 2.31 in Chapter 2 of
\cite{Karr91}. Detail will be done in the future work.
%
%
%
%
%
%

\subsubsection{Variance convergence order}

In the following we shall give some intuitive analysis of the variance
convergence order, which is not a rigorous proof. Hopefully we can use
extreme path theory to prove it thoroughly in the future work.

The variance of estimator $\hP_{\ell}^{(j)}-\hP_{\ell-1}^{(j)}\hR_{\ell}^{(j)}$ is
\begin{eqnarray*}
\VV[\hP_{\ell}-\hP_{\ell-1}\hR_{\ell}] &=&\VV[\hP_{\ell}-\hP_{\ell-1}+\hP_{\ell-1}(1-\hR_{\ell})] \\
&\leq &\left[ \left( \VV[\hP_{\ell}-\hP_{\ell-1}]\right) ^{\frac{1}{2}}+\left( %
\VV[\hP_{\ell-1}(1-\hR_{\ell})]\right) ^{\frac{1}{2}}\right] ^{2}.
\end{eqnarray*}

\bigskip The first part has the order of $O(h\dot{)}$. To see the order of $%
\hR_{\ell}$, let us examine the asymptomatic order of squared variance of two
components of $\hR_{\ell}$: the exponential part and product part alternatively:

\begin{equation*}
\hR_{\ell}=\exp \left[ \sum_{i=1}^{n_{\ell}}(\lambda _{i}^{f}-\lambda _{i}^{c})h_{i}%
\right] \prod\limits_{k}\dfrac{\lambda _{k}^{c}}{\lambda _{k}^{f}}.
\end{equation*}

where $\lambda _{i}^{f}=\lambda (\hS_{i}^{f},t_{i}),\lambda _{i}^{c}=\lambda
(\hS_{i}^{c},t_{i}).$

We will concentrate on the effect caused in midpoint interval.

\bigskip Let $h$ be the timestep of uniform fine grid. Given a midpoint
interval $[t_{j},t_{j+2}]$, since by definition $\lambda _{j+1}^{c}=\lambda
_{j}^{c}$, the exponential part for $[t_{j},t_{j+2}]$ is

\begin{eqnarray*}
\sum_{i=j}^{j+1}(\lambda _{i}^{f}-\lambda _{i}^{c})h_{i} &=&(\lambda
_{j}^{f}-\lambda _{j}^{c})(h_{j}+h_{j+1})+(\lambda _{j+1}^{f}-\lambda
_{j}^{f})h_{j+1} \\
&\sim& O(h)h+\frac{\partial \lambda }{\partial S}(\hS(t_{j})^f,t_{j})(\hS%
_{j+1}^{f}-\hS_{j}^{f}) h \\
&\sim& O(h^{3/2}).
\end{eqnarray*}

\bigskip The first term is implied by $O(h)$ strong convergence and second
one comes from $\hS_{j+1}^{f}-\hS_{j}^{f}\sim \hS_{j}^{f}\triangle
W_{h_{j+1}}\sim O(h^{1/2})$.

Summing all intervals in the grid, since there are $N=T/2h$ midpoint
intervals, and the rest contributes higher order, asymptotically we get $%
\sum_{i=1}^{n_{\ell}}(\lambda _{i}^{f}-\lambda _{i}^{c})h_{i}\sim O(h)$. Thus
the exponential part satisfies

\begin{equation*}
\exp \left[ \sum_{i=1}^{n_{\ell}}(\lambda _{i}^{f}-\lambda _{i}^{c})h_{i}\right]
\sim 1+O(h).
\end{equation*}

\bigskip For the product part, when $t_{j+2}$ is a jump time, by similar
argument it holds that :

\begin{eqnarray*}
\frac{\lambda _{j+1}^{f}}{\lambda _{j+1}^{c}} &=&1+\frac{\lambda
_{j}^{f}-\lambda _{j}^{c}+\lambda _{j+1}^{f}-\lambda _{j}^{f}}{\lambda
_{j}^{c}} \\
&\sim &1+O(h^{1/2}).
\end{eqnarray*}

So the product part satisfies
\begin{equation*}
\prod\limits_{k}\dfrac{\lambda _{k}^{c}}{\lambda _{k}^{f}}\sim 1+O(h^{1/2}).
\end{equation*}

Asymptotically we get

\begin{equation*}
\VV\left[ \hR_{\ell}\right] =O(h).
\end{equation*}

Therefore we have evidence to support that the variance convergence order
for the multilevel estimator should be

\begin{equation*}
\VV[\hP_{\ell}-\hP_{\ell-1}\hR_{\ell}]=O(h).
\end{equation*}

Although cumulative intensity method can deal with the case where jump rate
is not bounded, the variance convergence order $O(h)$ leads to a
computational complexity of $O(\varepsilon ^{-2}\left( \log \varepsilon
\right) ^{2})$.

\subsection{Thinning method}

The idea of the thinning method is to construct a Poisson process with a
constant rate $\lambda _{\sup }$ which is an upper bound of the
state-dependent rate. This gives a set of candidate jump times, and these
are then selected as true jump times with probability $\lambda
(S_{t},t)/\lambda _{\sup }$.

\subsubsection{Algorithm}

\bigskip Suppose that we have simulated the jump time grid $\mathbb{J}%
=\{\tau _{1},\tau _{2},\ldots ,\tau _{m}\}$ generated by a Poisson process
with constant rate $\lambda _{\sup }$, which includes times at which jumps
occur in $[0,T]$. On the other hand, consider a fixed time grid constituted
of $N$ timesteps, $t_{i}^{\prime }=i\times \frac{T}{N},~i=1,\ldots ,N,$
which is used in discretisation schemes for diffusive SDEs. Now the
superposition of them will be a jump-adapted thinning grid $\mathbb{T}%
=\{0=t_{0}<t_{1}<t_{2}<\ldots <t_{M}=T\}$.

For a process which we can simulate the exact increments we have the
following thinning procedure:

\begin{enumerate}
\item Generate the waiting time for next jump time $\tau_{i+1}-\tau_i$ from
a Poisson process with constant rate $\lambda_{\sup}$;

\item Simulate the evolution of the process up to time $\tau_{i+1}$;

\item Draw a uniform random number $U\sim[0,1]$,

\begin{enumerate}
\item
If $p=\dfrac{\lambda (S(\tau _{i+1}-),\tau _{i+1})}{\lambda _{\sup }}>U$,
accept $\tau _{i+1}$ as a real jump time and simulate the jump; otherwise go
to 2.
\end{enumerate}
\end{enumerate}

For the general processes we use certain discretisation scheme, e.g. we have
the following jump-adapted thinning (Milstein) scheme:

\begin{enumerate}
\item Generate the jump-adapted time grid for a Poisson process with
constant rate $\lambda_{\sup}$;

\item Simulate each timestep using the (Milstein) discretisation;

\item When the endpoint $t_{n+1}$ is a candidate jump time, generate a
uniform random number $U\sim \lbrack 0,1]$, and if $U<p_{\tau }=\dfrac{%
\lambda (S(\tau -),\tau )}{\lambda _{\sup }}$, then accept $t_{n+1}$ as a
real jump time and simulate the jump.
\end{enumerate}

\bigskip If we look from the perspective of random measure notation, the
thinning method can also be formulated in the following ways. First let us
recall some definitions of point processes and random measures.

\subsubsection{Point processes and random measures}

There are two kinds of basic stochastic processes. The path evolution of the
first kind is driven by continuous increment, while in the second kind the
path will only change in the certain jump times. In order to describe the
stochastic processes of discrete paths, we need to introduce point process.

Given a filtered probability space $(\Omega ,\mathcal{F}_{t},\mathbb{P})$,
if we have a increasing sequence of increasing stopping times%
\begin{equation*}
0=T_{0}<T_{1}<T_{2}<\ldots
\end{equation*}

and
\begin{equation*}
\lim_{n\rightarrow \infty }T_{n}=\infty ,
\end{equation*}

then the point (counting) process $N_{t}$ associated with stopping times
(jump times) is defined as

\begin{equation*}
N_{t}=\sum_{n\geq 1}1_{\{T_{n}<t\}}.
\end{equation*}

It counts the number of jumps up to time $t$. To depict the jump amplitude
(mark) at each jump time, we can define marked point processes and
associated random measure.

Suppose mark $Y_{n}$ are random variables taking values on a mark space $%
E\subseteq R^{r}\backslash \left\{ 0\right\} $, and $Y_{n}$ are $\mathcal{F}
_{T_{n}}$ measurable, then $(Y_{n},T_{n})$ is called a marked point process
on \ $E\times \lbrack 0,\infty ].$

For any $A\subseteq B(E)$ and any $\omega \in \Omega $,

\begin{equation*}
\mu (\omega ;[0,t),A):=\sum_{n\geq 1}1_{\{T_{n}\left( \omega \right)
<t\}}1_{\{Y_{n}\left( \omega \right) \in A\}}
\end{equation*}

\bigskip\ defines the random measure associated with marked point process $%
(Y_{n},T_{n})$. It counts the number of jumps within $[0,t)$ whose amplitude
belonging to $A$. The class of marked point process is quite general,
actually it is a bijection to the class of c\`{a}dl\`{a}g process.

For each $\omega \in \Omega $, $\mu \left( \omega ;\cdot ,\cdot \right) $ is
an Radon measure on $\left( E\times \lbrack 0,\infty ],\mathcal{\ B}\left(
E\times \lbrack 0,\infty ]\right) \right) .$ Thus for each $\mathbb{P}$%
-measuable \ function $f$, the integral

\begin{equation*}
\int_{E\times \lbrack 0,t]}f(\omega ;z,s)\mu (\omega ;\D z,\D s):=\sum
_{\substack{ n\geq 1  \\ T_{n}\left( \omega \right) <t}}f\left( Y_{n}\left(
\omega \right) ,T_{n}\left( \omega \right) \right)
\end{equation*}

is a well-defined random variable.

\bigskip The compensator $\varphi (\D z)\D s$ of the random measure is
defined so that for all bounded $\mathbb{P}$-measurable \ function $f,$ we
have

\begin{equation*}
\int_{0}^{t}\int_{z\in E}f(z,s)\mu (\D z,\D s)-\int_{0}^{t}\int_{z\in
E}f(z,s)\varphi (\D z)\D s
\end{equation*}

is a martingale with respect to $\mathcal{F}_{t}$.

\bigskip Following the assumption of Platen, we assume that $\int_{z\in
E}\varphi (\D z)<\infty $ for all $s>0$.

\bigskip Under those notations, (\ref{eq:jumpSDE}) can be rewritten  as:

\begin{equation*}
\D S(t)=a(S(t-),t)\D t+b(S(t-),t)\D W(t)+c(S(t-),t)\int_{z\in
E}(z-1)p_{\lambda }(\D z,\D t),\quad 0\leq t\leq T.
\end{equation*}%
where $p_{\lambda }\left( \omega ;\cdot ,\cdot \right) $ is a Poisson random
measure, which means the compensator $\varphi (\D z)=\lambda g(z)\D z$ is
time independent and $g(z)$ is p.d.f of the mark . In this case the waiting
time $T_{n+1}-T_{n}$ is exponential distributed with a parameter $\lambda $.

Those general definitions are to allow more flexibility of the dynamics of
SDEs, e.g. it can admit state-dependent intensity. The dynamics of the
state-dependent jump-diffusion SDEs we will deal with can be written as

\begin{equation}
\D S(t)=a(S(t-),t)\D t+b(S(t-),t)\D W(t)+\int_{z\in E}c(S(t-),t,z)\mu (\D z,%
\D t),\quad 0\leq t\leq T.  \label{eq:jumpsde}
\end{equation}%
The compensator of $\mu $ is defined to be $\varphi (S(t-),\D %
z)\lambda (S(t-),t)g(z)\D z.$ We also adopt the assumption in \cite{GM04}
that it is bounded by a constant $\lambda _{\sup }$ and is absolutely
continuous.

The jump-adapted thinning Milstein scheme for (\ref{eq:jumpsde}) can be
formulated as

\begin{eqnarray*}
\hS_{n+1}^{-} &=&\hS_{n}+a_{n}\,h_{n}+b_{n}\,\Delta W_{n}+{\textstyle\frac{1%
}{2}}\,b_{n}^{\prime }b_{n}\,(\Delta W_{n}^{2}-h_{n}), \\
\hS_{n+1} &=&\hS_{n+1}^{-}+\int_{z\in E}1_{\{\frac{\lambda (\hS%
_{n+1}^{-},t_{n+1})}{\lambda _{\sup }}>U_{i}\}}c(\hS_{n+1}^{-},t_{n+1},z)\mu
(\D z,t_{n+1}).
\end{eqnarray*}

%
%
%
%

\bigskip %
%

%

\subsubsection{Multilevel treatment}

In the multilevel implementation, if we use the above algorithm with
different acceptance probabilities for fine and coarse level, there may be
some samples in which a jump candidate is accepted for the fine path, but
not for the coarse path, or vice versa. Because of first order strong
convergence, the difference in acceptance probabilities will be $O(h)$, and
hence there is an $O(h)$ probability of coarse and fine paths differing in
accepting candidate jumps. Such differences will give an $O(1)$ difference
in the payoff value, and hence the multilevel variance will be $O(h)$. A
more detailed analysis of this is given in \cite{xg11b}.

To improve the variance convergence rate, we use a change of measure so that
the acceptance probability is the same for both fine and coarse paths. This
is achieved by taking the expectation with respect to a new measure $Q$:
\begin{equation*}
\EE[\hP_{\ell}-\hP_{\ell-1}] = \EE_{Q}[\hP_{\ell}\prod_{\tau }R_{\tau }^{f}-%
\hP_{\ell-1}\prod_{\tau }R_{\tau }^{c}]
\end{equation*}%
where $\tau $ are the jump times. The acceptance probability for a candidate
jump under the measure $Q$ is defined to be $\frac{1}{2}$ for both coarse
and fine paths, instead of $p_\tau = \lambda (S(\tau -),\tau ) \,/\,
\lambda_{\sup }$. The corresponding Radon-Nikodym derivatives are
\begin{equation*}
R_{\tau }^{f}=\left\{ \begin{aligned} &2p_{\tau }^{f}, ~&\mbox{if}~
U<\frac{1}{2}&~;\\ &2(1-p_{\tau }^{f}),~&\mbox{if}~ U\geq\frac{1}{2}&~,
\end{aligned}\right. \quad \quad R_{\tau }^{c}=\left\{ \begin{aligned}
&2p_{\tau }^{c}, ~&\mbox{if}~ U<\frac{1}{2}&~;\\ &2(1-p_{\tau
}^{c}),~&\mbox{if}~ U\geq\frac{1}{2}&~, \end{aligned}\right.
\end{equation*}
Since $R_{\tau }^{f}-R_{\tau }^{c} = O(h)$ and $\hP_{\ell}-\hP_{\ell-1} =
O(h)$, this results in the multilevel correction variance $\VV_Q[\hP%
_{\ell}\prod_{\tau }R_{\tau }^{f} -\hP_{\ell-1}\prod_{\tau }R_{\tau }^{c}]$
being $O(h^{2})$.

\begin{figure}[t!]
\begin{center}
\includegraphics[width=\textwidth]{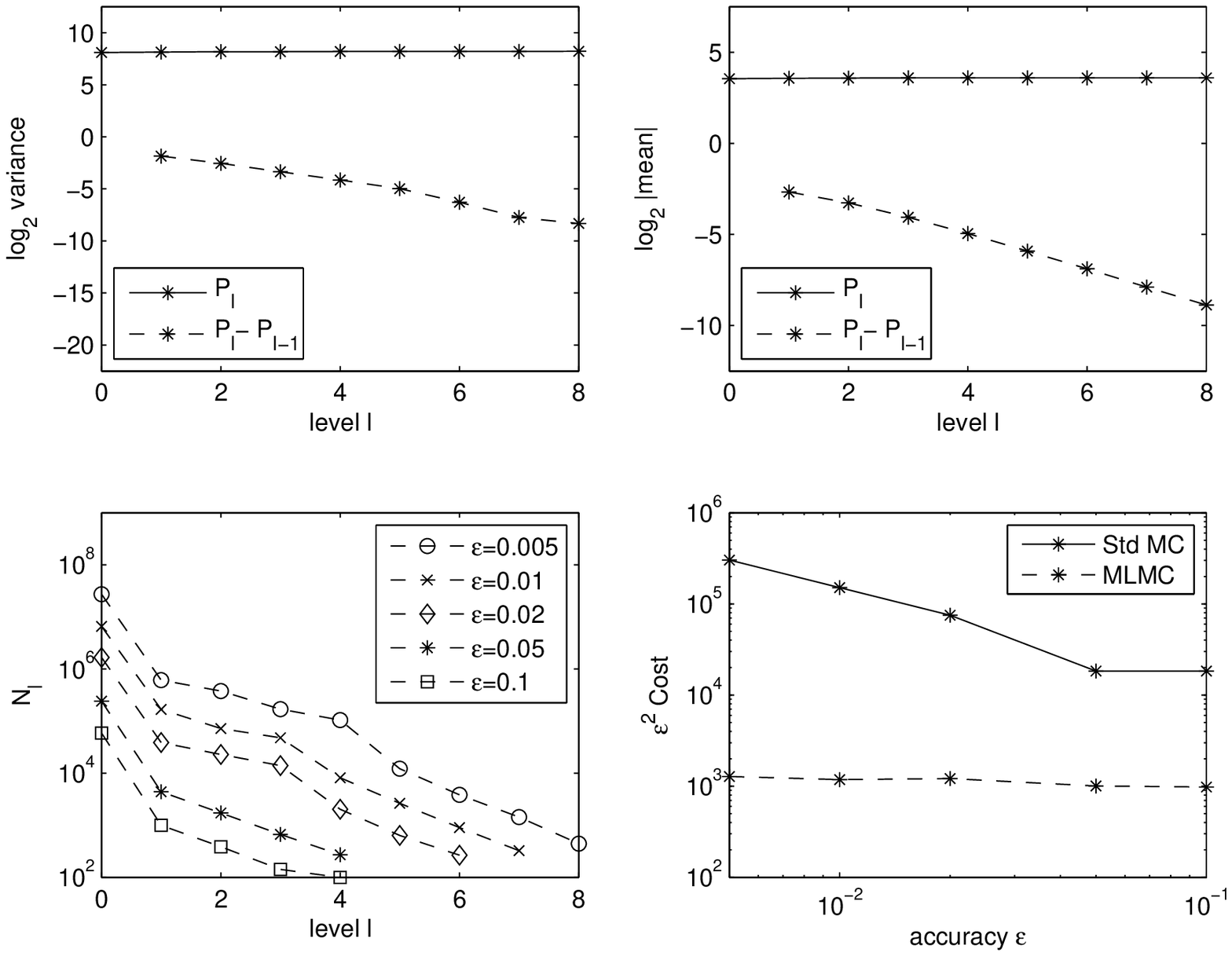}
\end{center}
\par
\vspace{-0.2in}
\caption{European call option with path-dependent Poisson rate using
thinning without a change of measure}
\label{fig:state2}
\end{figure}

\begin{figure}[t!]
\begin{center}
\includegraphics[width=\textwidth]{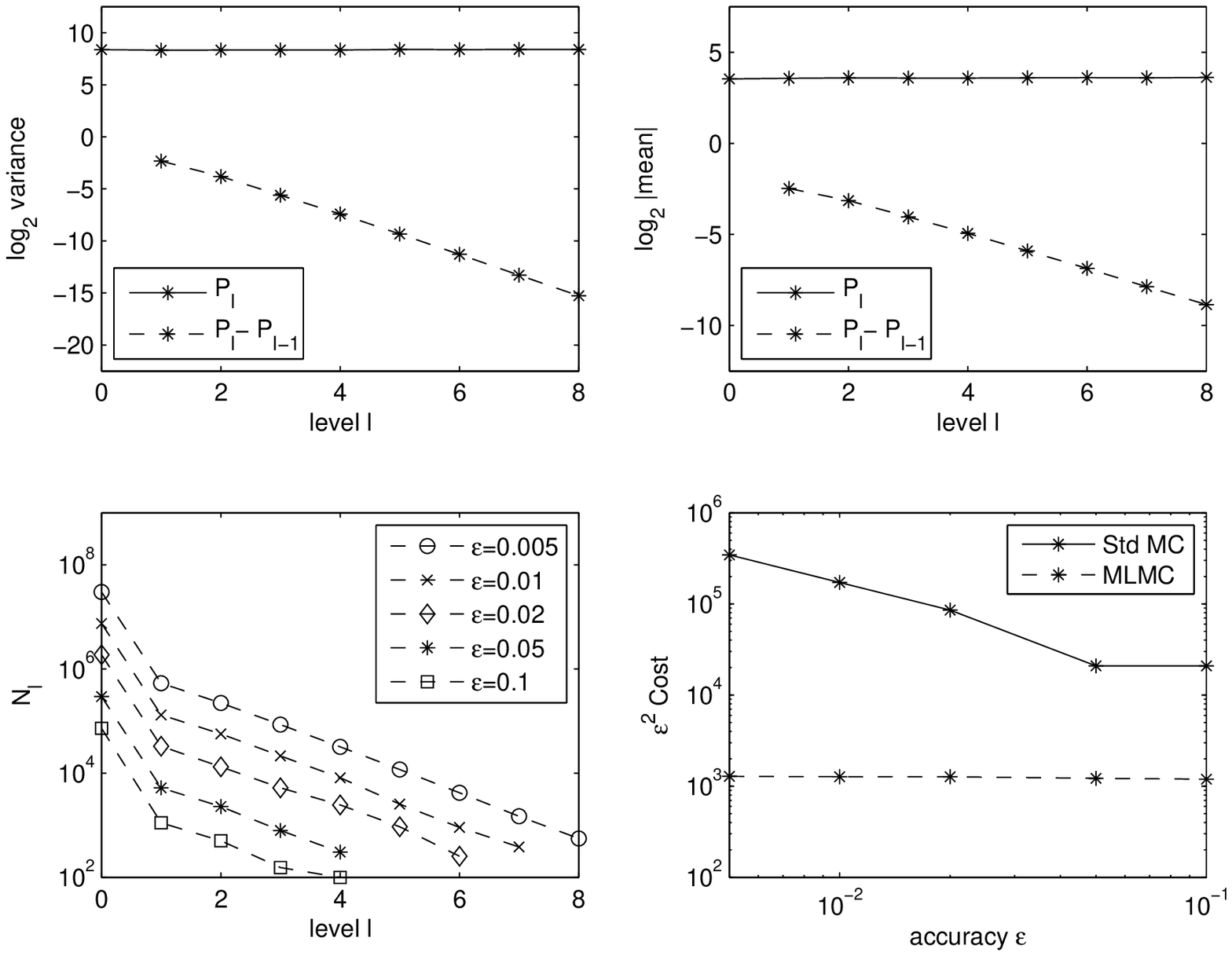}
\end{center}
\par
\vspace{-0.2in}
\caption{European call option with path-dependent Poisson rate using
thinning with a change of measure}
\label{fig:state3}
\end{figure}

\begin{figure}[t!]
\begin{center}
\includegraphics[width=\textwidth]{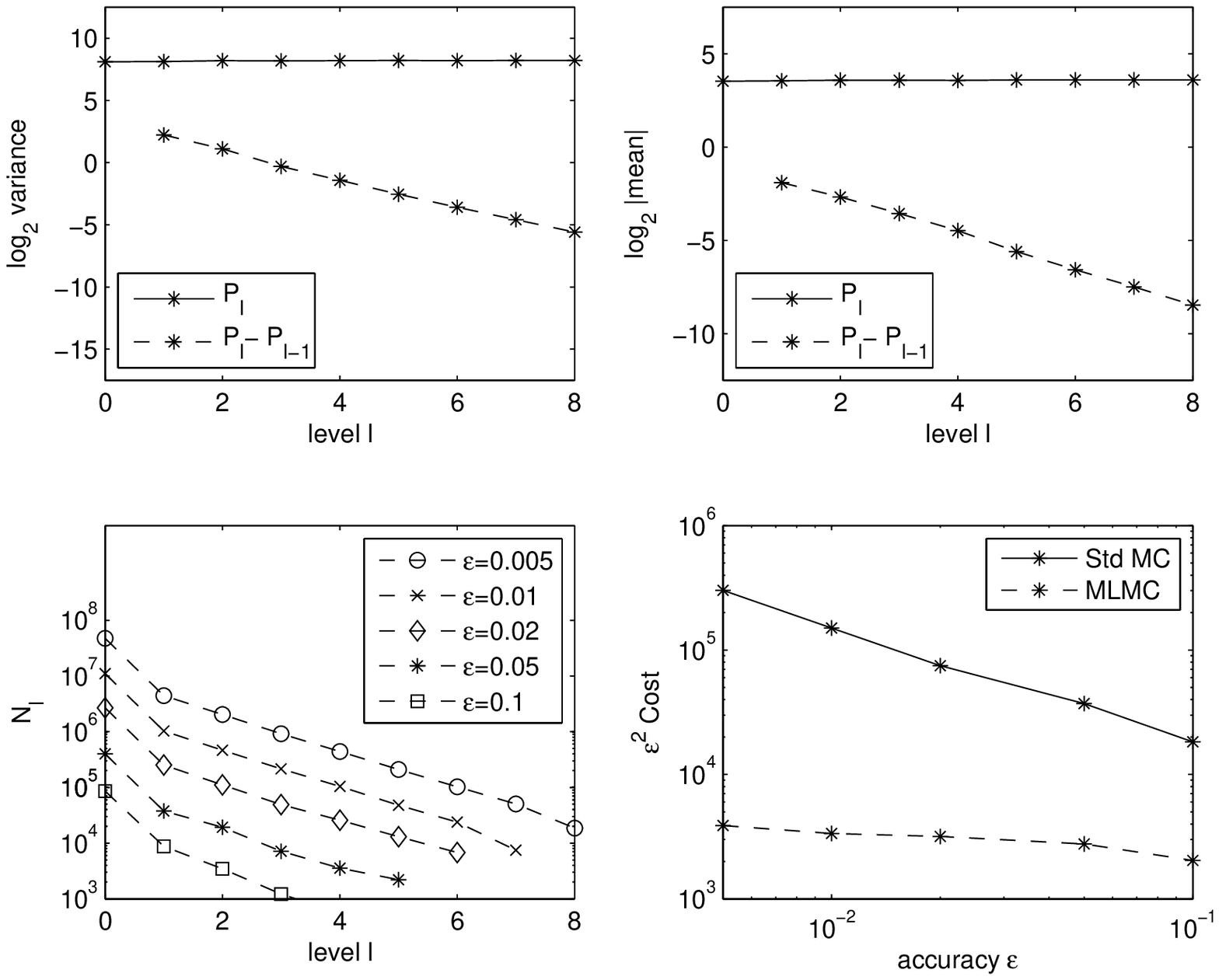}
\end{center}
\par
\vspace{-0.2in}
\caption{European call option with path-dependent Poisson rate using
cumulative intensity method}
\label{fig:cum}
\end{figure}

\bigskip The weak convergence of the jump-adapted discretisation with the
thinning procedure is proved in \cite{gm04b}, for a class of payoffs on
which they impose to be $4$th differentiable and that its up to 4th order
partial derivatives are uniformly bounded. By Stone--Weierstrass theorem we
can construct a sequence of smoothing polynomials which uniformly converges
to continuous payoff. The limits of this approach is that it is invalid for
discontinuous payoffs. To prove the convergent variance of multilevel
estimator, assuming the Lipschitz condition on $\lambda $, we decompose the
estimator into the constant rate part and the Randon-Nikdym derivative part,
disentangling the effect of path-dependence of intensity from the estimator.
Under such assumption, we can obtain the weak convergence of the estimators
for various payoffs by the same decomposition, circumventing the
difficulties caused by discontinuous payoffs. The advantage of this argument
is that it can reduce the analysis to the constant rate case so that the
proof is simplified.

If the analytic formulation is expressed using the same thinning and change
of measure, the weak error can be decomposed into two terms as follows:
\begin{equation*}
\EE_{Q}\left[ \hP_{\ell }\prod_{\tau }R_{\tau }^{f}-P\prod_{\tau }R_{\tau }%
\right] \ =\ \EE_{Q}\left[ (\hP_{\ell }-P)\ \prod_{\tau }R_{\tau }^{f}\right]
\ +\ \EE_{Q}\left[ P\ (\prod_{\tau }R_{\tau }^{f}-\prod_{\tau }R_{\tau })%
\right] .
\end{equation*}%
Using H{\"{o}}lder's inequality, the bound $\max (R_{\tau },R_{\tau
}^{f})\leq 2$ and standard results for a Poisson process, the first term can
be bounded using weak convergence results for the constant rate process, and
the second term can be bounded using the corresponding strong convergence
results \cite{xg11b}. This guarantees that the multilevel procedure does
converge to the correct value.

\subsubsection{Numerical results}

We show numerical results for a European call option using the underlying
dynamics under risk-neutral measure:

\begin{equation*}
\dfrac{\D S(t)}{S(t-)}=r\,\D t+\sig\,\D W(t)++\int_{z\in E}z\mu (\D z,\D %
t)-\int_{z\in E}zf(z)\D z\lambda \D t,\quad 0\leq t\leq T,
\end{equation*}

where the random measure $\mu $ has compensator $\lambda \D t$ with $\lambda
=\frac{1}{1+(S(t-)/S_{0})^{2}}$. $\ $The mark has a $\log $ normal
distribution the density function of which is denoted by $f(z)$.

We use $\lambda _{\sup }=1$ to generate thinning process. All other
parameters as used previously for the constant rate cases.

Comparing Figures \ref{fig:state2} and \ref{fig:state3} we see that the
variance convergence rate is significantly improved by the change of
measure, but there is little change in the computational cost. This is due
to the main computational effort being on the coarsest level, which suggests
using quasi-Monte Carlo on that level \cite{gw09}.

The bottom left plot in Figure \ref{fig:state2} shows a slightly erratic
behaviour. This is because the $O(h_{\ell })$ variance is due to a small
fraction of the paths having an $O(1)$ value for $\hP_{\ell }-\hP_{\ell -1}$%
. In the numerical procedure, the variance is estimated using an initial
sample of 100 paths. When the variance is dominated by a few outliers, this
sample size is not sufficient to provide an accurate estimate, leading to
this variability.

For comparison, we also show the numerical result given by cumulative
intensity method  in Figure \ref{fig:cum}. The bottom left plot indicates
the $O(h_{\ell })$ variance and the rest plots can be understood
consequently.

\section{Conclusions}

In this work we extend the Multilevel approach to scalar jump-diffusion SDEs
using jump-adapted schemes. The second order variance convergence is
maintained in the constant rate case, by constructing estimators using a
previous Brownian interpolation technique. In the state-dependent rate case,
we use thinning with a change of measure to avoid the asynchronous jumps in
the fine and coarse levels. We have also investigated an alternative
approach which can handle cases in which there is no upper bound on the jump
rate. 

The first natural future work is to do rigorous numerical analysis on the
convergence of variance of correction terms, and weak convergence of the ML
estimators, which is work in progress \cite{xg11b}. The second direction is to investigate other cases of model
based on specific infinite activity L\'{e}vy processes, e.g. variance gamma. We also plan to
investigate whether the multilevel quasi-Monte Carlo method will further
reduce the cost.

%

\bibliographystyle{alpha}

\end{document}